\documentclass{llncs}

\usepackage{url}
\usepackage{amssymb}
\usepackage{amsmath}
\usepackage{alltt}
\usepackage{mathtools}
\usepackage{dsfont}
\usepackage[T1]{fontenc}
\usepackage{fancybox}

\newcommand\mfbox[1]{\,\,\fbox{#1}\,\,}
\newcommand\K{$\mathds K$}
\newcommand\exec{\textnormal{\,\guillemotright\,}}
\newcommand\thenelse[2]{\textnormal{\bf then~}#1\textnormal{\bf~else~}#2\textnormal{\bf~end}}
\newcommand\loopend[1]{\textnormal{\bf loop~}#1\textnormal{\bf~end}}
\newcommand\create[1]{\textnormal{\bf create~}#1}
\newcommand\forget[1]{\textnormal{\bf forget~}#1}
\newcommand\call[1]{\textnormal{\bf call}\,#1}
\newcommand\assgn{\,:\!=\,}
\newcommand\lasso[2]{\textnormal{lasso(}#1,#2\textnormal{)}}
\newcommand\reg[2]{\textnormal{reg(}#1,#2\textnormal{)}}

\newcommand\cellS[2]{\small \langle\,\, #1 \,\,\rangle_{\textnormal{#2}}}
\newcommand\cellSL[2]{\small \langle\,\, #1  \,\,\ldots\rangle_{\textnormal{#2}}}
\newcommand\cellSR[2]{\small \langle\ldots\,\,  #1 \,\,\rangle_{\textnormal{#2}}}

\newcommand\cellU[3]
{
\small
\begin{tabular}{r@{}c@{}l}
$\langle\,\,$ & $#1$ & $\,\,\rangle_{\textnormal{#3}}$\\
\cline{2-2}
& $#2$ &\\
\end{tabular}
}
\newcommand\cellUL[3]
{
\small
\begin{tabular}{r@{}c@{}l}
$\langle\,\,$ & $#1$ & $ \,\,\ldots\rangle_{\textnormal{#3}}$\\
\cline{2-2}
& $#2$ &\\
\end{tabular}
}
\newcommand\cellUR[3]
{
\small
\begin{tabular}{r@{}c@{}l}
$\langle\ldots\,\, $ & $#1$ & $\,\,\rangle_{\textnormal{#3}}$\\
\cline{2-2}
& $#2$ &\\
\end{tabular}
}

\mathcode`:="003A  
\mathcode`;="003B  
\mathcode`?="003F  
\mathcode`|="026A  
\mathcode`<="4268  
\mathcode`>="5269  

\mathchardef\ls="213C    
\mathchardef\gr="213E    

\newenvironment{todo}{\bigskip\hrule\medskip\noindent}{\medskip\hrule\bigskip}

\begin{document}
\bibliographystyle{abbrv}

\title{Expression-based aliasing for OO--languages}

\author{Georgiana Caltais\inst{1}}
\authorrunning{ }

\institute{
Department of Computer Science, ETH Z\"urich, Switzerland
}

\maketitle 

{\begin{abstract}
Alias analysis has been an interesting research topic in verification and optimization of programs. The undecidability of determining whether two expressions in a program may reference to the same object is the main source of the challenges raised in alias analysis.
%
In this paper we propose an extension of a previously introduced alias calculus based on program expressions, to the setting of unbounded program executions such as infinite loops and recursive calls.
Moreover, we devise a corresponding executable specification in the {\K}-framework.
An important property of our extension is that, in a non-concurrent setting, the corresponding alias expressions can be over-approximated in terms of a notion of regular expressions. This further enables us to show that the associated {\K}-machinery implements an algorithm that always stops and provides a sound over-approximation of the ``may aliasing'' information, where soundness stands for the lack of false negatives.
As a case study, we analyze the integration and further applications of the alias calculus in SCOOP. The latter is an object-oriented programming model for concurrency, recently formalized in Maude; {\K} definitions can be compiled into Maude for execution.
\end{abstract}
}

\section{Introduction}

A research direction of interest in Computer Science is the application of \emph{alias analysis} in verification and optimization of programs. One of the challenges along this line of research has been the undecidability of determining whether two expressions in a program \emph{may} reference the same object. A rich suite of  approaches aiming at providing a satisfactory balance between scalability and precision has already been developed in this regard. Examples include:
 (i) intra-procedural frameworks~\cite{Landi:1991:PAP:99583.99599,Landi:1992:USA:161494.161501} that handle isolated functions only, and their inter-procedural counterparts~\cite{Landi:1992:USA:161494.161501,Myers:1981:PID:567532.567556,Hind:1999:IPA:325478.325519} that consider the interactions between function calls;
(ii) {type-based} techniques~\cite{Diwan:1998:TAA:277652.277670};
(iii) flow-based techniques~\cite{Burke-flow-ins,Choi:1993:EFI:158511.158639} that establish aliases depending on the control-flow information of a procedure;
(iv) context-(in)sensitive approaches~\cite{Emami:1994:CIP:178243.178264,Wilson:1995:ECP:207110.207111} that depend on whether the calling context of a function is taken into account or not;
(v) field-(in)sensitive approaches~\cite{Mine:2006:FVA:1134650.1134659,Albert:2009:FVA:1693345.1693376} that depend on whether the individual fields of objects in a program are traced or not.
More details on such classifications can be found in~\cite{DBLP:conf/cpp/RobertL12}, for instance. For a comprehensive survey on alias analyses for object-oriented programs, corresponding issues and remaining open problems, we refer the interested reader to the works in~\cite{DBLP:series/lncs/SridharanCDFY13,DBLP:conf/paste/Hind01}.

Of particular interest for the work in this paper is the untyped, flow-sensitive, field sensitive, inter-procedural and context-sensitive calculus for \emph{may aliasing}, introduced in~\cite{Meyer-aliasing-13}. The aforementioned calculus covers most of the aspects of a modern object-oriented language, namely: object creation and deletion, conditionals, assignments, loops and (possibly recursive) function calls. The approach in~\cite{Meyer-aliasing-13} abstracts the aliasing information in terms of explicit access paths~\cite{DBLP:conf/pldi/LarusH88} referred to as \emph{alias expressions}. Consider, for an example, the code
\begin{equation}
\begin{array}{l}
\label{eq:code-loop-end}
x \,:= y;\\
\loopend{x \,:= x.next}
\end{array}
\end{equation}
The corresponding execution causes $x$ to become aliased to $y.next.next.\,\ldots$, with a possibly infinite number of occurrences of the field $next$. The set of associated alias expressions can be equivalently written as:
\begin{equation}
\label{eq:intor-ex}
\{[x,\,y.next^k] \mid k \geq 0\}.
\end{equation}
The sources of imprecision introduced by the calculus in~\cite{Meyer-aliasing-13} are limited to ignoring tests in conditionals, and to ``cutting at length $L$'', for the case of possibly infinite alias relation as in~(\ref{eq:intor-ex}). Intuitively, the cutting technique considers sequences longer than a given length $L$ as aliased to all expressions.

There is a huge literature on heap analysis for aliasing~\cite{DBLP:conf/paste/Hind01}, but hardly any paper that presents a calculus as  in~\cite{Meyer-aliasing-13} allowing the derivation of alias relations as the result of applying various instructions of a programming language.

Our focus is two folded. First, we want extend the framework in~\cite{Meyer-aliasing-13} to the setting of unbounded program executions such as infinite loops and recursive calls. In accordance, the goal is to provide a way to shift from ``finite'' to ``infinite behaviours''. This can be achieved in a rather straightforward manner, by redefining the construct $\loopend{p}$ in~\cite{Meyer-aliasing-13} according to the informal semantics: ``execute $p$ repeatedly any number of times, including zero''.
However,
developing a corresponding mechanism for reasoning on ``may aliasing'' in a finite number of steps is not trivial. The key observation that paves the way to a possible (finite state-based) modeling in a non-concurrent setting is that the alias expressions corresponding to loops and recursive calls grow in a regular fashion. Hence, they are finitely representable, as it is easy to see in~(\ref{eq:intor-ex}), for instance. Such regularities cannot be exploited in concurrent contexts, due to the ``non-determinism'' of process interaction.

A similar technique exploiting regular behaviour of (non-concurrent) programs, in order to reason on ``may aliasing'', was previously introduced in~\cite{Asav-aliasing-k}. In short, the results in~\cite{Asav-aliasing-k} utilize {abstract} representations of programs in terms of finite pushdown systems, for which infinite execution paths have a regular structure (or are ``lasso shaped'')~\cite{DBLP:conf/concur/BouajjaniEM97}. Then, in the style of abstract interpretation~\cite{DBLP:journals/jlp/CousotC92}, the collecting semantics is applied over the (finite state) pushdown systems to obtain the alias analysis itself.
In short, the main difference with the results in~\cite{Asav-aliasing-k} consists in how the abstract memory addresses corresponding to pointer variables are represented. In~\cite{Asav-aliasing-k} these  range over a finite set of natural numbers. In this paper we consider alias expressions build according to the calculus in~\cite{Meyer-aliasing-13}.

The work in~\cite{Asav-aliasing-k} also proposes an implementation of pushdown systems in the {\K}-framework~\cite{DBLP:journals/jlp/RosuS10}. The latter is an executable semantic framework based on Rewriting Logic (RL)~\cite{DBLP:conf/fct/MeseguerR11}, and has successfully been used for defining programming languages and corresponding formal analysis tools. Moreover, {\K} definitions have a direct implementation in {\K}-Maude~\cite{DBLP:conf/wrla/SerbanutaR10}.

We agree that it could be worth presenting our analysis as an abstract interpretation (AI)~\cite{DBLP:journals/jlp/CousotC92}. A modelling exploiting the machinery of AI (based on abstract domains, abstraction and concretization functions, Galois connections, fixed-points, {\emph etc.}) is an interesting, but different research topic per se. 

Our second interest w.r.t. may aliasing is its integration in SCOOP~\cite{DBLP:conf/acsd/MorandiSNM13} -- a simple object oriented programming model for concurrency; thus an operational based approach on handling the alias calculus is more appropriate. The basis of a RL-based framework for the design and analysis of the SCOOP model was recently set in~\cite{DBLP:conf/acsd/MorandiSNM13}. 
The reference implementation of SCOOP is Eiffel~\cite{DBLP:books/ph/Meyer91}. 
The integration of alias analysis belongs to a more ambitious goal, namely, the construction of a RL-based toolbox for the analysis of SCOOP programs (examples include a deadlock detector and a type checker).

\paragraph{Our contribution.}
{
By drawing inspiration from, and building on top of the results in~\cite{Meyer-aliasing-13,Asav-aliasing-k}, in this paper we propose:
\begin{itemize}\itemsep3pt
\item an extension of the (finite) alias calculus in~\cite{Meyer-aliasing-13} to the setting of unbounded program executions, and a sound over-approximation technique based on ``regular alias expressions'', for non-concurrent settings;
\item a RL-based specification of the extended calculus;
\item an algorithm that always terminates and provides a sound over-approximation of ``may aliasing'' by exploiting a notion of regular (finitely representable) aliases, for non-concurrent settings.
\end{itemize}
Moreover, we analyze the integration, implementation and further applications of the alias calculus in SCOOP.
}

\paragraph{{Paper structure.}}{
The paper is organized as follows. In Section~\ref{sec:alias-calc} we introduce the extension of the alias calculus in~\cite{Meyer-aliasing-13} to unbounded executions. In Section~\ref{sec:implem-k} we provide the RL-based executable specification of the calculus in the {\K} semantic framework. The implementation in SCOOP, and further applications are discussed in Section~\ref{sec:alias-SCOOP}. In Section~\ref{sec:conclusions} we draw the conclusions and provide pointers to future work. 
}

\section{The alias calculus}
\label{sec:alias-calc}

In this section we define an extension of the calculus in~\cite{Meyer-aliasing-13}, to unbounded program executions. Moreover, based on the idea behind the \emph{pumping lemma for regular languages}~\cite{Rabin:1959:FAD:1661907.1661909}, we devise a corresponding sound over-approximation of ``may aliasing'' in terms of regular expressions, applicable in sequential contexts. This paves the way to developing an algorithm for the aliasing problem, as presented in Section~\ref{sec:implem-k}, in the formal setting of the {\K} semantic framework~\cite{DBLP:journals/jlp/RosuS10}.

\paragraph{\bf Preliminaries.}{
We proceed by briefly recalling the notion of \emph{alias relation} and a series of associated notations and basic operations, as introduced in~\cite{Meyer-aliasing-13}.

We call an \emph{expression} a (possibly infinite) path of shape $x.y.z.\,\ldots$, where $x$ is a local variable, class attribute or {\it{Current}}, and $y, z, \ldots$ are attributes. Here, {\it{Current}}, also known as {\it{this}} or {\it{self}}, stands for the current object.
For an arbitrary alias expression $e$, it holds that $e . {\it Current} = {\it Current} . e = e$. 
Let $E$ represent the set of all expressions of a program. An \emph{alias relation} is a symmetric and irreflexive binary relation over $E \times E$.

Given an alias relation $r$ and an expression $e$, we define
\[
r/e = \{e\} \cup \{x:\, E \mid [x,e] \in r\}
\]
denoting the set consisting of all elements in $r$ which are aliased to $e$, plus $e$ itself.

Let $x$ be an expression; we write $r\,-\,x$ to represent $r$ without the pairs with one element of shape $x.e$.

We say that an alias relation is \emph{dot complete} whenever for any $t, u,v$ and $a$ it holds that if $[t,u]$ and $[t.a,v]$ are alias pairs, then $[u.a, v]$ is an alias pair and, moreover, if $a$ is in the domain of $t$, then $[t.a, u.a]$ is an alias pair.
By the ``domain of $t$'' we refer to a method or a field in the class corresponding to the object referred by the expression associated to $t$. For instance, given a class NODE with a field $next$ of type NODE, and a NODE object $x$, we say that $next$ is in the domain of $t \,=\, x.next.next$.
For the sake of brevity, we write {\it dot-complete}$(r)$ for the closure under dot-completeness of a relation $r$.

The notation $r[x=u]$ represents the relation $r$ augmented with pairs $[x,y]$ and made dot complete, where $y$ is an element of $u$.
}

\subsection{Extension to unbounded executions}
\label{sec:ext-inf-expr}

We further introduce an extension of the alias calculus in~\cite{Meyer-aliasing-13} to infinite alias relations corresponding to unbounded executions such as infinite loops or recursive calls. The main difference in our approach is reflected by the definition of loops, which now complies to the usual fixed-point denotational semantics.

The alias calculus is defined by a set of axioms ``describing'' how the execution a program affects the aliasing between expressions. As in~\cite{Meyer-aliasing-13}, the calculus ignores tests in conditionals and loops. The \emph{program instructions} are defined as follows:
\begin{equation}
\label{eq:BNF-control-struct}
\begin{array}{r c l}
p & ::= & p \,;\, p \mid \thenelse{p}{p} \mid\\
& & \create{x} \mid \forget{x} \mid t \assgn s \mid\\
&& \loopend{p} \mid \call{f(l)} \mid x.\call{f(l)}.
\end{array}
\end{equation}
In short, we write
$
r \exec p
$
to represent the alias information obtained by executing $p$ when starting with the initial alias relation $r$.

The axiom for sequential composition is defined in the obvious way:
\begin{equation}
\label{eq:def-seq-comp}
r \exec (p\,;\,q) = (r \exec p) \exec q.
\end{equation}

Conditionals are handled by considering the union of the alias pairs resulted from the execution of the instructions corresponding to each of the two branches, when starting with the same initial relation:
\begin{equation}
\label{eq:def-then-else}
r \exec (\thenelse{p}{q})  = r \exec p\,\, \cup\,\, r \exec q.
\end{equation}

As previously mentioned,  we define $r \exec \loopend{p}$ according to its informal semantics : ``execute $p$ repeatedly any number of times, including zero''. The corresponding rule is:
\begin{equation}
\label{eq:def-loop}
r \exec (\loopend{p}) = \bigcup_{n \in \mathds{N}} (r \exec p^n)
\end{equation}
where $\cup$ stands for the union of alias relations, as above.
This way, our calculus is extended to infinite alias relations. 
This is the main difference with the approach in~\cite{Meyer-aliasing-13} that proposes a ``cutting'' technique restricting the model to a maximum length $L$.  In~\cite{Meyer-aliasing-13}, sequences longer than $L$ are considered as aliased to all expressions. Orthogonally, for sequential settings, we provide finite representations of infinite alias relations based on over-approximating regular expressions, as we shall see in Section~\ref{sec:sound-over-approx}. 

Both the creation and the deletion of an object $x$ eliminate from the current alias relation all the pairs having one element prefixed by $x$:
\begin{equation}
\label{eq:def-create-delete}
\begin{array}{rcl}
r \exec (\create{x}) & = & r - x\\
r \exec (\forget{x}) & = & r - x.
\end{array}
\end{equation}

The (qualified) function calls comply to their initial definitions in~\cite{Meyer-aliasing-13}: 
\begin{equation}
\label{eq:def-qualified-call}
\begin{array}{rcl}
r \exec (\call{f(l)}) & = & (r[f^\bullet:l])\exec \mid f \mid\\
r \exec (x.\call{f(l)}) & = & x.((x'.r) \exec \call{f(x'.l)}).
\end{array}
\end{equation}
Here $f^\bullet$ and $\mid f \mid$ stand for the formal argument list and the body of $f$, respectively, whereas $r[u:v]$ is the relation $r$ in which every element of the list $v$ is replaced by its counterpart in $u$. Intuitively, the negative variable $x'$ is meant to transpose the context of the qualified call to the context of the caller. Note that ``$.$'' ({\it i.e.}, the constructor for alias expressions) is generalized to distribute over lists and relations:
$x.[a,b,\ldots] = [x.a, x.b, \ldots]$.

For an example, consider a class $C$ in an OO-language, and an associated procedure $f$ that assigns a local variable $y$, defined as: $f(x) \,\{ \,\, y \assgn x \,\, \}$.
Then, for instance, the aliasing for $a.\call{f(a)}$ computes as follows:
\[
\begin{array}{rc}
\emptyset \,\,\exec\,\, a.\call{f(a)} & = \\
a.(a'.\emptyset \,\,\exec\,\, y\,:=a'.a) & =\\
a.(\emptyset \,\,\exec\,\, y\,:=\textnormal{\it Current}) & =\\
\textnormal{\it dot-complete}(\{[a.y, a]\}).
\end{array}
\]

Recursive function calls can lead to infinite alias relations. In sequential settings, as for the case of loops, the mechanism exploiting sound regular over-approximations in order to derive finite representations of such relations is presented in the subsequent sections.

The axiom for assignment is as well in accordance with its original counterpart in~\cite{Meyer-aliasing-13}:
\begin{equation}
\label{eq:def-assign}
\begin{array}{rcl}
r \exec (t \assgn s) & = & {\textnormal{\bf given~}} r_{1} = r[ot = t]\\
&& {\textnormal{\bf then~}} (r_{1} - t)[t\,=\,(r_1 \slash s \,-\,t)] - ot  {\textnormal{\bf~end}}\\
\end{array}
\end{equation}
where $ot$ is a fresh variable (that stands for ``old $t$'').
Intuitively, the aliasing information w.r.t. the initial value of $t$ is ``saved'' by associating $t$ and $ot$ in $r$ and closing the new relation under dot-completeness, in $r_1$. Then, the initial $t$ is ``forgotten'' by computing $r_1 - t$ and the new aliasing information is added in a consistent way. Namely, we add all pairs $(t, s')$, where $s'$ ranges over {$r_1 \slash s \,-\, t$} representing all expressions already aliased with $s$ in $r_1$, including $s$ itself, but without $t$. Recall that alias relations are not reflexive, thus by eliminating $t$ we make sure we do not include pairs of shape $[t, t]$. Then, we consider again the closure under dot-completeness and forget the aliasing information w.r.t. the initial value of $t$, by removing $ot$.

\begin{remark}
It is worth discussing the reason behind \emph{not} considering transitive alias relations.
Assume the following program: 
\[
\thenelse{x \assgn y}{y \assgn z}
\]
Based on the equations~(\ref{eq:def-then-else}) and~(\ref{eq:def-assign}) handling conditionals and assignments, respectively, the calculus correctly identifies the alias set: $\{[x, y], [y, z]\}$. Including $[x, z]$ would be semantically equivalent to the execution of the two branches in the conditional at the same time, which is not what we want.
\end{remark}


\subsection{A sound over-approximation}
\label{sec:sound-over-approx}

\label{rm:unfolding}
In a sequential setting, the challenge of computing the alias information in the context of (infinite) loops and recursive calls reduces to evaluating their corresponding ``unfoldings'', captured by expressions of shape
\[
r \exec p^{\omega},
\]
with $\omega$ ranging over naturals plus infinity, r an (initial) alias relation ($r = \emptyset$), and $p$ a \emph{basic control block} defined by:
\begin{equation}
\label{eq:BNF-basic-struct}
\begin{array}{r c l}
p & ::= & p \,;\, p \mid \thenelse{p}{p} \mid\\
& & \create{x} \mid \forget{x} \mid\\
&& t \assgn s.
\end{array}
\end{equation}
The value $r \exec p^{\omega}$ refers to the alias relation obtained by recursively executing the control block $p$, and it is calculated in the expected way:
\[
\begin{array}{rcl}
r \exec p^{0} & = & r\\
r \exec p^{k+1} & = & (r \exec p^{k}) \exec p.
\end{array}
\]

Consider again the code in~(\ref{eq:code-loop-end}):
\[
\begin{array}{l}
x \assgn y;\\
\loopend{x \assgn x.next}.
\end{array}
\]
Its execution generates the alias relation
\[
(((\emptyset \exec (x \assgn y)) \exec (x \assgn x.next)) \exec (x \assgn x.next) \ldots
\]
including an infinite number of pairs of shape: 
\begin{equation}
\label{eq:inf-rel-next}
[x, y.next],\, [x, y.next.next],\, [x, y.next.next.next] \ldots~~.
\end{equation}
A similar reasoning does not hold for concurrent applications, where process interaction is not ``regular''.

In what follows we provide a way to compute finite representations of infinite alias relations in sequential settings.
The key observation is that alias expressions corresponding to unbounded program executions grow in a regular fashion. See, for instance, the aliases in~(\ref{eq:inf-rel-next}), which are pairs of type $[x, y.next^{k \geq 1}]$.

Regular expressions are defined similarly to the regular languages over an alphabet.
We say that an expression is \emph{regular} if it is a local variable, class attribute or {\it{Current}}. Moreover, the concatenation $e_1\,.\,e_2$ of two regular expressions $e_1$ and $e_2$ is also regular. Given a regular alias expression $e$, the expression $e^*$ is also regular; here $(-)^*$ denotes the Kleene star~\cite{Kleene56}. We call an alias relation \emph{regular} if it consists of pairs of regular expressions.

\begin{lemma}
\label{lm:sequential-regularity}
Assume $p$ a program built according to the rules in~(\ref{eq:BNF-control-struct}).
Then, in a sequential setting, the relation $\emptyset \exec p$ is regular.
\end{lemma}

\begin{proof}
The result follows by induction on the structure of $p$. We refer to Appendix~\ref{sec:reg-expr} for the detailed proof.
\end{proof}

Inspired by  the idea behind the \emph{pumping lemma for regular languages}~\cite{Rabin:1959:FAD:1661907.1661909}, we define a \emph{lasso} property for alias relations, which identifies the repetitive patterns within the structure of the corresponding alias expressions.
The intuition is that such patterns will occur for an infinite number of times due to the execution of loops or recursive function calls.
Then, we supply sound over-approximations of ``lasso'' relations, based on regular alias expressions.

In the context of alias relations, we say that the lasso property is satisfied by $r$ and $r'$ whenever the following two conditions hold: (1) $r$ behaves like a \emph{lasso base} of $r'$. Namely, all the pairs $[e_1, e_2] \in r$ are used to generate  elements $[e'_1, e'_2] \in r'$, by repeating tails of prefixes of $e_1$ and $e_2$, respectively, and (2) $r'$ is a \emph{lasso extension} of $r$. Namely, all the pairs in $r'$ are generated from elements of $r$ by repeating tails of their prefixes.
For example, if $e_1$ above is an expression of shape $x.y.z.w$, then $e'_1$ can be $x.y.y.z.w$ if we consider the tail $y$ of the prefix $x.y$, or $x.y.z.y.z.w$ if we take the tail $y.z$ of the prefix $x.y.z$.

Formally, consider $r$ and $r'$ two alias relations, and $x_i, y_i$ and $z_i$ a set of (possibly empty) expressions, for $i \in \{1,2\}$. Then:
\begin{equation}
\label{eq:def-lasso}
\lasso{r}{r'} = 
([x_1 y_1 z_1, x_2 y_2 z_2] \in r \textnormal{~~iff~~} [x_1 y_1 y_1 z_1, x_2 y_2 y_2 z_2] \in r').
\end{equation}
For the simplicity of notation we sometimes omit the dot-separators between expressions. For instance, we write $x\,y\,z$ in lieu of $x.y.z$.

Assuming a lasso over $r$ and $r'$, we compute a relation consisting of regular expressions over-approximating $r$ and $r'$ as:
\begin{equation}
\label{eq:def-reg}
\begin{array}{rcl}
\reg{r}{r'} & = & \{[x_1 y_1^* z_1, x_2 y_2^* z_2] \,\mid\\
&& \,\,\, [x_1 y_1 z_1, x_2 y_2 z_2] \in r\,\land\\
&& \,\,\, [x_1 y_1 y_1 z_1, x_2 y_2 y_2 z_2] \in r'\} 
\end{array}
\end{equation}
where $x_i, y_i$ and $z_i$ are possibly empty expressions, for $i \in \{1,2\}$. As previously indicated, the over-approximation is sound w.r.t. the repeated application of a basic control block as in~(\ref{eq:BNF-basic-struct}), in the way that it does not introduce any false negatives:

\begin{lemma}
\label{lm:reg-expr}
Consider $r$ and $r'$ two alias relations, and $p$ a basic control block in a sequential setting. If $r\exec p = r'$ and $\lasso{r}{r'} = true$, then the following holds for all $n\geq 1$:
\[
r \exec p^{n} \in \reg{r}{r'}.
\]
\end{lemma}
\begin{proof}
The reasoning is by induction on $n$. The base case follows immediately, whereas the induction step is proved by ``reductio ad absurdum''. A detailed proof is included in Appendix~\ref{sec:soundness}.
\qed
\end{proof}

\section{A {\K}-machinery for collecting aliases}
\label{sec:implem-k}

In this section we provide the specification of a RL-based mechanism collecting the alias information in the {\K} semantic framework~\cite{DBLP:journals/jlp/RosuS10}. We choose {\K} more as a notational convention to enable compact and modular definitions. In reality, the {\K}-rules in this section are implemented in Maude, as rewriting theories, on top of the formalization of SCOOP~\cite{DBLP:conf/acsd/MorandiSNM13} (we refer to Section~\ref{sec:alias-SCOOP} for more details on our approach).

In short, our strategy is to start with a program built on top of the control structures in~(\ref{eq:BNF-control-struct}), then to apply the corresponding {\K}-rules in order to get the ``may aliasing'' information in a designated {\K}-cell ($\cellS{-}{\textnormal{al}}$). Independently of the setting (sequential or concurrent) one can exploit this approach in order to evaluate the aliases of a given finite length $L$. We also show that for sequential contexts, the application of the {\K}-rules is finite and the aliases in the final configuration soundly over-approximate the (infinite) ``may alias'' relations of the calculus.

\begin{paragraph}{\bf Brief overview of {\K}.}{
{\K}~\cite{DBLP:journals/jlp/RosuS10} is an executable semantic framework based on Rewriting Logic~\cite{DBLP:conf/fct/MeseguerR11}. It is suitable for defining (concurrent) languages and corresponding formal analysis tools, with straightforward implementation in {\K}-Maude~\cite{DBLP:conf/wrla/SerbanutaR10}. {\K}-definitions make use of the so-called \emph{cells}, which are labelled and can be nested, and (rewriting) \emph{rules} describing the intended (operational) semantics.

A \emph{cell} is denoted by $\cellS{-}{\textnormal{[name]}}$, where [name] stands for the \emph{name of the cell}. A construction $\cellS{.}{\textnormal{n}}$ stands for an \emph{empty cell} named n. We use ``pattern matching'' and write $\cellSL{c}{\textnormal{n}}$ for a cell with content $c$ at the top, followed by an arbitrary content ($\ldots$). Orthogonally, we can utilize cells of shape $\cellSR{c}{\textnormal{n}}$ and $\cellS{\ldots c \ldots}{{\textnormal{n}}}$, defined in the obvious way.

Of particular interest is $\cellS{-}{\textnormal{k}}$ -- the \emph{continuation cell}, or the \emph{$k$-cell}, holding the stack of program instructions (associated to one processor), in the context of a programming language formalization. We write
\[
\cellSL{i_1  \curvearrowright i_2}{\textnormal{k}}
\]
for a set of instructions to be ``executed'', starting with instruction $i_1$, followed by $i_2$. The associative operation $ \curvearrowright$ is the instruction sequencing.

A {\K}-rewrite rule
\begin{equation}
\label{eq:def-k-rew-rule-ex}
\cellSL{c}{\textnormal{$n_1$}}
\cellS{c'}{\textnormal{$n_2$}}
~\Rightarrow~
\cellSL{c'}{\textnormal{$n_1$}}
\cellSR{c'}{\textnormal{$n_3$}}
\end{equation}
reads as: if cell \textnormal{${n_1}$} has $c$ at the top and cell \textnormal{${n_2}$} contains value $c'$, 
then $c$ is replaced by $c'$ in \textnormal{${n_1}$} and $c'$ is added at the end of the cell \textnormal{${n_3}$}. The content of \textnormal{${n_2}$} remains unchanged.
}
In short,~(\ref{eq:def-k-rew-rule-ex}) is written in a {\K}-like syntax as:
\[
\cellUL{c}{c'}{\textnormal{$n_1$}}
\cellS{c'}{\textnormal{$n_2$}}
\cellUR{.}{c'}{\textnormal{$n_3$}}.
\]
\end{paragraph}

We further provide the details behind the {\K}-specification of the alias calculus. As expected, the $k$-cell retains the instruction stack of the object-oriented program.
We utilize cells
$
\langle - \rangle_{\textnormal{al}}
$
to enclose the current alias information, and the so-called \emph{back-tracking cells} 
$
\langle - \rangle_{\textnormal{bkt-\ldots}}
$
enabling the sound computation of aliases for the case of \mbox{$\thenelse{\!\!-\!\!}{\!\!-\!\!}$} and, in non-concurrent contexts, for loops and (possibly recursive) function calls.
As a convention, we mark with ($\clubsuit$) the rules that are sound only for non-concurrent applications, based on Lemma~\ref{lm:reg-expr}.
Due to space limitations, in what follows we introduce only the {\K}-rules for handling assignments and loops. The entire specification is included in Appendix~\ref{sec:calculus-complete-spec}.

As expected, the assignment rule simply restores the current alias relation according to its axiom in~(\ref{eq:def-assign}), and removes the assignment instruction from the top of the $k$-cell:

\begin{equation}
\label{eq:k-assign}
\cellU{r}{(r_{1} - t)[t\,=\,(r_1 \slash s \,-\,t)] - ot}{al}\cellUL{t \assgn s}{.}{k}\,\,\,\,\,
\textnormal{with~} r_{1} = r[ot = t]
\end{equation}

For $\loopend{p}$, we utilize a meta-construction $p \mfbox{l} \loopend{p}$ simulating the unfolding corresponding to~(\ref{eq:def-loop}), and a back-tracking stack $\langle - \rangle_{\textnormal{bkt-l}}$ collecting the alias information obtained after each execution of $p$. Moreover, the {\K}-implementation exploits the result in Lemma~\ref{lm:reg-expr}.
Whenever a ``lasso'' is reached, the infinite rewriting is prevented by resuming the infinite application of $p$ in terms of a sound over-approximating alias relation. The {\K}-rules are as follows.

First, the aforementioned unfolding is performed, and the alias relation before $p$ is stored in the back-tracking cell as $\langle r \rangle_{\textnormal{al-o}} \langle p \rangle_{\textnormal{l}}$:

\begin{equation}
\label{eq:k-loop-whole}
\cellS{r}{al}
\cellUL{\loopend{p}}{p \mfbox{l} \loopend{p}}{k}
\cellUL{.}{\cellS{r}{al-o}\cellS{p}{l}}{bkt-l}
\end{equation}

If the alias relation $r'$ obtained after the successful execution of $p$ (marked by $\mfbox{l}$ at the top of the continuation) is not a lasso of the aliasing $r$ before $p$ (previously stored in $\langle -\rangle_{\textnormal{bkt-l}}$) then $p$ is constrained to a new execution by becoming the top of the $k$-cell, and $r'$ is memorized for back-tracking:

\begin{equation}
\label{eq:k-loop-not-lasso}
\cellS{r'}{al}
\cellUL{\mfbox{l} \loopend{p}}{p \mfbox{l} \loopend{p}}{k}
\cellUL{\cellS{r}{al-o}\cellS{p}{l}}{\cellS{r'}{al-o}\cellS{p}{l}}{bkt-l}
\textnormal{~if not~} \lasso{r}{r'}~~(\clubsuit)
\end{equation}

Last, if a lasso is reached after the execution of $p$, then the current aliasing is soundly replaced by a ``regular'' over-approximation $\reg{r}{r'}$, the corresponding back-tracking information is removed from $\langle -\rangle_{\textnormal{bkt-l}}$ and the {\bf loop} instruction is eliminated from the $k$-cell:

\begin{equation}
\label{eq:k-loop-lasso}
\cellU{r'}{\reg{r}{r'}}{al}\!\!
\cellUL{\mfbox{l} \loopend{p}}{.}{k}\!\!
\cellUL{\cellS{r}{al-o}\cellS{p}{l}}{.}{bkt-l}
\textnormal{~if~} \lasso{r}{r'}~~(\clubsuit)
\end{equation}

In a non-concurrent setting, the machinery orchestrating the {\K}-rules introduced in this section, and thoroughly discussed in Appendix~\ref{sec:calculus-complete-spec}, implements an algorithm that always terminates and provides a sound over-approximation of ``may aliasing''.

\begin{theorem}
\label{th:dec-proc}
Consider $p$ a program built on top of the control structures in~(\ref{eq:BNF-control-struct}), that executes in a sequential setting. Then, the application of the corresponding {\K}-rules when starting with $p$ and an empty alias relation, is a finite rewriting of shape
\[
\cellS{\emptyset}{\textnormal{al}}\cellS{p}{\textnormal{k}}\,\,\xRightarrow{(*)}\,\,\cellS{r}{\textnormal{al}}\cellS{.}{\textnormal{k}},
\]
with $r$ a sound over-approximation of the aliasing information corresponding to the execution of $p$.
\end{theorem}
\begin{proof}
The key observation is that, due to the execution of loops and/or recursive calls, expressions can infinitely grow in a \emph{regular} fashion. Hence, a lasso is always reached. Consequently, the control structure generating the infinite behaviour is removed from the $k$-cell, according to the associated {\K}-specification for loops and/or recursive calls. This guarantees termination. Moreover, recall that the regular expressions replacing the current alias information are a sound over-approximation, according to Lemma~\ref{lm:reg-expr}.
\qed
\end{proof}

Observe that the $RL$-based machinery can simulate precisely the ``cutting at length L'' technique in~\cite{Meyer-aliasing-13}. It suffices to disable the rules ($\clubsuit$) and stop the rewriting after L steps.

The naturalness of applying the resulted aliasing framework is illustrated in the example in Appendix~\ref{sec:example-k-machinery}, for the case of two mutually recursive functions.

\section{Integration in SCOOP}
\label{sec:alias-SCOOP}

In this section we provide a brief overview on the integration and applicability of the alias calculus in SCOOP~\cite{DBLP:conf/acsd/MorandiSNM13} -- a simple object-oriented programming model for concurrency.
Two main characteristics make SCOOP simple: 1) just one keyword programmers have to learn and use in order to enable concurrent executions, namely, {\it{separate}} and 2) the burden of orchestrating concurrent executions is handled within the model, therefore reducing the risk of correctness issues. 

In short, the key idea of SCOOP is to associate to each object a processor, or {\it handler} (that can be a CPU, or it can also be implemented in software, as a process or thread). Assume a processor $p$ that performs a call $o.f()$ on an object $o$. If $o$ is declared as ``separate'', then $p$ sends a request for executing $f()$ to $q$ -- the handler of $o$ (note that $p$ and $q$ can coincide). Meanwhile, $p$ can continue. Processors communicate via {\it channels}.

The Maude semantics of SCOOP in~\cite{DBLP:conf/acsd/MorandiSNM13} is defined over tuples of shape   
\[
\langle p_1 \,::\, St_{1} \mid \ldots \mid p_n \,::\, St_{n}, \sigma \rangle
\]
where, $p_i$ denotes a processor (for $i \in \{1, \ldots, n\}$), $St_i$ is the call stack of $p_i$ and $\sigma$ is the {\it state} of the system. States hold the information about the {\it heap} (which is a mapping of references to objects) and the {\it store} (which includes formal arguments, local variables, {\it etc}.).

The assignment instruction, for instance, is formally specified as the transition rule:
\begin{equation}
\label{eq:assign-Morandi}
\dfrac{\textnormal{a is fresh}}{
\Gamma \vdash \langle p \,::\, t\,:=s;\, St,\, \sigma \rangle \rightarrow
\langle p\,::\, \textnormal{eval}(a, s);\, \textnormal{wait}(a);\, \textnormal{write}(t, a.data);\,St,\, \sigma\rangle
}
\end{equation}
where, intuitively, ``eval$(a, s)$'' evaluates $s$ and puts the result on channel $a$, ``wait$(a)$'' enables processor $p$ to use the evaluation result, ``write$(t, a.data)$'' sets the value of $t$ to $a.data$, $St$ is a call stack, and $\Gamma$ is a typing environment~\cite{nienaltowski2007practical} containing the class hierarchy of a program and all the type definitions.

At this point it is easy to understand that the {\K}-rule for assignments
\[
\cellU{r}{(r_{1} - t)[t\,=\,(r_1 \slash s \,-\,t)] - ot}{al}\cellUL{t \assgn s}{.}{k}\,\,\,\,\,
\textnormal{with~} r_{1} = r[ot = t]~~~~(\ref{eq:k-assign})
\]
can be straightforwardly integrated in~(\ref{eq:assign-Morandi}) by enriching the state structure with a new field encapsulating the alias information, and considering instead the transition $\Gamma \vdash \langle p \,::\,\, t\, {:=} s;\, St,\, \sigma \rangle \rightarrow
\langle p\,::\, \textnormal{eval}(a, s);\, \textnormal{wait}(a);\, \textnormal{write}(t, a.data);\,St,\, \sigma'\rangle$
where 
\[
\begin{array}{lr}
\sigma.aliases = r ~~~~~~~~~
\sigma'.aliases =(r_{1} - t)[t\,=\,(r_1 \slash s \,-\,t)] - ot 
\end{array}
\]
with $r$ and $r_1$ as in~(\ref{eq:k-assign}).
The integration of all the {\K}-rules of the alias calculus on top of the Maude formalization of SCOOP can be achieved by following a similar approach.

For a case study, one can download the SCOOP formalization at:\\
{\url{https://dl.dropboxusercontent.com/u/1356725/SCOOP.zip}}\\
and run the command\\
{\small \verb+> maude SCOOP.maude ..\examples\aliasing-linked_list.maude+}\\
corresponding to the code in~(\ref{eq:code-loop-end}):
\[
\begin{array}{l}
x \,:= y;~~
\loopend{x \,:= x.next}.
\end{array}
\]
The console outputs the aliased expressions for a rewriting of depth $100$ which include, as expected, pairs of shape $[x,\, y.next^k]$. (The over-approximating mechanism for sequential settings is still to be implemented.)

As can be observed based on the code in {\small \verb+aliasing-linked_list.maude+}, in order to implement our applications in Maude, we use intermediate (still intuitive) representations.
For instance, the class structure defining a node in a simple linked list, with filed \emph{next} is declared as:
{
\small
\begin{verbatim}
class 'NODE
    create {'make}
    ( attribute { 'ANY } 'next : [?, . , 'NODE] ; )
   [...]
end ;
\end{verbatim}
}
\noindent
where {\small \verb+'next : [?, . , 'NODE]+} stands for an object of type NODE, that is handled by the current processor ({\small \verb+.+}) and that can be Void ({\small \verb+?+}), and {\small \verb+'make+} plays the role of a constructor. The intermediate representation of the instruction block in~(\ref{eq:code-loop-end}) is:
{
\small
\begin{verbatim}
assign ('x, 'y);
until False loop ( assign ('x, 'x . 'next(nil)) ; ) end ;
\end{verbatim}
}
We include in Appendix~\ref{app:Maude-aliasing} the whole class structure corresponding to~(\ref{eq:code-loop-end}), together with (the relevant parts of) the console output.
For a detailed description of SCOOP and its Maude formalization we refer the interested reader to the work in~\cite{DBLP:conf/acsd/MorandiSNM13}.

\subsection{Further applications of the alias calculus}
\label{sec:further-applications}

Apart from providing an alias analysis tool, the alias calculus can be exploited in order to build an abstract semantics of SCOOP. For example, an abstraction of the assignment rule~(\ref{eq:k-assign}) would omit the evaluation of the right-hand side of the assignment $t\,:=s$ and the associated message passing between channels:
\begin{equation}
\notag
\dfrac{\cdot}{
\Gamma \vdash \langle p \,::\, t\,:=s;\, St,\, \sigma \rangle \rightarrow
\langle p\,::\, \,St,\, \sigma'\rangle
}
\end{equation}
where 
\[
\begin{array}{lr}
\sigma.aliases = r & ~~~~~~~~
\sigma'.aliases =(r_{1} - t)[t\,=\,(r_1 \slash s \,-\,t)] - ot 
\end{array}
\]
with $r$ and $r_1$ as in~(\ref{eq:k-assign}).
This way one derives a simplified, reduced semantics of SCOOP, more appropriate for model checking, for instance; the current SCOOP formalization in Maude is often too large for this purpose.
A survey on abstracting techniques on top of Maude executable semantics is provided in~\cite{DBLP:conf/fct/MeseguerR11}.  

Furthermore, the aliasing information could be used for the so-called ``deadlocking'' problem, where two or more executing threads are each waiting for the other to finish. In the context of SCOOP, this is equivalent to identifying whether a set of processors reserve each other circularly ({\it i.e.}, there is a Coffman deadlock). This situation might occur, for instance, in a Dinning Philosophers scenario, where both philosophers and forks are objects residing on their own processors. The difficulty of identifying such deadlocks stems from the fact that SCOOP processors are known from object references, which \emph{may be aliased}.

\section{Conclusions}
\label{sec:conclusions}

In this paper we provide an extension of the alias calculus in~\cite{Meyer-aliasing-13} from finite alias relations to infinite ones corresponding to loops and recursive calls. Moreover, we devise an associated executable specification in the {\K} semantic framework~\cite{DBLP:journals/jlp/RosuS10}. In Theorem~\ref{th:dec-proc} we show that the RL-based machinery implements an algorithm that always terminates with a sound over-approximation of ``may aliasing'', in non-concurrent settings. This is achieved based on the sound (finitely representable) over-approximation of (``lasso shaped'') alias expressions in terms of regular expressions, as in Lemma~\ref{lm:reg-expr}. 
We also discuss the integration and applicability of the alias calculus on top of the Maude formalization of SCOOP~\cite{DBLP:conf/acsd/MorandiSNM13}. 

\medskip
An immediate direction for future work is to identify interesting (industrial) case studies to be analyzed using the framework developed in this paper.
We are also interested in devising heuristics comparing the efficiency and the precision ({\it e.g.}, the number of false positives introduced by the alias approximations) between our approach and other aliasing techniques. In this respect, we anticipate that the rewriting modulo associativity, together with the pattern matching capabilities of Maude will accelerate the identification of the ``lasso'' properties and the corresponding over-approximating regular alias expressions. This could eventually provide an effective reasoning apparatus for the ``may aliasing'' problem.

Another research direction is to derive alias-based abstractions for analyzing concurrent programs. We foresee possible connections with the work in~\cite{DBLP:conf/concur/HoareMSW09} on \emph{concurrent Kleene algebra} formalizing choice, iteration, sequential and concurrent composition of programs. The corresponding definitions exploit abstractions of programs in terms of traces of events that can depend on each other. Thus, obvious challenges in this respect include: (i) defining notions of dependence for all the program constructs in this paper, (ii) relating the concurrent Kleene operators to the semantics of the SCOOP concurrency model and (iii) checking whether fixed-points approximating the aliasing information can be identified via fixed-point theorems.

Furthermore, it would be worth investigating whether the graph-based model of alias relations introduced in~\cite{Meyer-aliasing-13} can be exploited in order to derive finite {\K} specifications of the extended alias calculus. In case of a positive answer, the general aim is to study whether this type of representation increases the speed of the reasoning mechanism, and why not -- its accuracy. With the same purpose, we refer to a possible integration with the technique in~\cite{DBLP:conf/pldi/ChaseWZ90} that handles point-to graphs via a stack-based algorithm for fixed-point computations.

We are also interested to what extent an abstract semantics based on aliases for SCOOP can be exploited for building more efficient analysis tools such as deadlock detectors, for instance. A survey on similar techniques that abstract away from possibly irrelevant information w.r.t. the problem under consideration is provided in~\cite{DBLP:conf/fct/MeseguerR11}.  

\paragraph{Acknowledgements}{
We are grateful for valuable comments to
M\u ariuca As\u avoae,   Alexander Kogtenkov, Jos\'e Meseguer, Bertrand Meyer, Benjamin Morandi and Sergey Velder.
The research leading to these results has received funding from the
European Research Council under the European Union's Seventh Framework
Programme (FP7/2007-2013) / ERC Grant agreement no. 291389.
}




\newpage
\appendix

\section{Regular expressions in sequential settings}
\label{sec:reg-expr}
In this section we provide the proof of Lemma~\ref{lm:sequential-regularity}; we proceed by demonstrating a series of intermediate results.

\begin{remark}
\label{rm:op-preserve-reg}
We first observe that the operations $r \slash s$, $r - x$, dot-completeness and $r[x = u]$ introduced in Section~\ref{sec:alias-calc} preserve the regularity of an alias relation $r$. 
\end{remark}

Then, we define a notion of \emph{finite execution} control blocks:
\begin{equation}
\label{eq:BNF-basic-struct-finite}
\begin{array}{r c l}
p & ::= & \create{x} \mid \forget{x} \mid t \assgn s\mid \\
& & p \,;\, p \mid \thenelse{p}{p} \mid\\
& & \call{f(l)} \mid x . \call{f(l)}
\end{array}
\end{equation}
where $f$ stands for a non-recursive function.

It is easy to see that the execution of control blocks as in~(\ref{eq:BNF-basic-struct-finite}) preserve the regularity of alias relations as well.

\begin{lemma}
\label{lm:reg-fin-exec}
For all regular alias relations $r$ and $p$ a finite-execution control block, in a sequential setting, it holds that $r \exec p$ is also regular.
\end{lemma}
\begin{proof}
The proof follows immediately, by induction on the structure of $p$ and Remark~\ref{rm:op-preserve-reg}.
Base cases are: $\create{x}$, $\forget{x}$ and $t \assgn s$. For function calls, the result is a consequence of their corresponding unfolding, based on the definitions in~(\ref{eq:def-qualified-call}).
\end{proof}

\begin{remark}
\label{rm:rec-calls-loop}
W.r.t. may aliasing, recursive calls can be handled via loops.
Consider, for instance the recursive function
\[
f(x)\,\, \{\, B_1;\,\, f(y);\,\, B_2 \,\}
\]
where $B_1$ and $B_2$ are instruction blocks built as in~(\ref{eq:BNF-control-struct}). It is intuitive to see that computing the may aliases resulted from the execution of $f(x)$ reduces executing unfoldings of shape:
\[
\loopend{B_1};\,\,\loopend{B_2}.
\] 
\end{remark}

Moreover, unbounded program executions also preserve regularity.
\begin{lemma}
\label{lm:reg-infin-exec}
For all regular alias relations $r$ and $p$ a control block that can execute unboundedly, in a sequential setting, it holds that $r \exec p$ is also regular.
\end{lemma}
\begin{proof}
The proof follows by induction on the number of nested loops in $p$ and Remark~\ref{rm:rec-calls-loop}.
\end{proof}

Then, the result in Lemma~\ref{lm:sequential-regularity} follows immediately by Lemma~\ref{lm:reg-fin-exec} and Lemma~\ref{lm:reg-infin-exec}.

\section{Sound over-approximations}
\label{sec:soundness}

In what follows we provide the proof of Lemma~\ref{lm:reg-expr}:\\
\emph{
Consider $r$ and $r'$ two alias relations, and $p$ a basic control block. If $r\exec p = r'$ and $\lasso{r}{r'} = true$, then the following holds for all $n \geq 1$:
\[
r \exec p^{n} \in \reg{r}{r'}.
\]
}

\begin{proof}
We proceed by induction on $n$.
\begin{itemize}
\item {\it Base case}: $n=1$. By hypothesis it holds that $\lasso{r}{r'} = true$. Hence, according to the definition of $\lasso{-}{-}$ in~(\ref{eq:def-lasso}), there exists a one-to-one correspondence of the shape
\[
[x_1 y_1 z_1, x_2 y_2 z_2] \in r \textnormal{~~iff~~} [x_1 y_1 y_1 z_1, x_2 y_2 y_2 z_2] \in r'
\]
between the elements of $r$ and $r'$, respectively.

Consequently, by the definition of $\reg{-}{-}$ in~(\ref{eq:def-reg}), it is easy to see that 
\[r' \in \reg{r}{r'}.\]

\item {\it Induction step.} Fix a natural number $n$ and suppose that
\begin{equation}
\label{eq:ind-step}
r \exec p^k \in \reg{r}{r'}
\end{equation}
for all $k \in \{1, \ldots, n\}$. We want to prove that~(\ref{eq:ind-step}) holds also for $k = n+1$.

We continue by ``reductio ad absurdum''. Consider 
\[
\overline{r} = r \exec p^n \in \reg{r}{r'},
\]
and assume that
\begin{equation}
\label{eq:red-absurdum}
\overline{r} \exec p \not \in \reg{r}{r'}
\end{equation}
Clearly, the execution of $p$ when starting with $\overline{r}$ identifies an alias pair which is not in $\reg{r}{r'}$. Given that $p$ is a basic control block as in~(\ref{eq:BNF-basic-struct}), and based on the corresponding definitions in~(\ref{eq:def-seq-comp})--(\ref{eq:def-assign}), it is not difficult to observe that the regular structure of the alias information can only be broken via a new added pair $(t, s')$ associated to an assignment $t \assgn s$ within $p$.

Let $p = C[t \assgn s]$, where $C$ is a context built according to~(\ref{eq:BNF-basic-struct}), and $t \assgn s$ is the upper-most assignment instruction in the syntactic tree associated to $p$, that introduces a pair $[t,s']$ which is not in $\reg{r}{r'}$. Assume that $\tilde{r}$ is the intermediate alias relation obtained by reducing $\overline{r} \exec C[t\assgn s]$ according to the equations~(\ref{eq:def-seq-comp})--(\ref{eq:def-assign}), before the application of the assignment axiom corresponding to $t\assgn s$.

Note that $t \assgn s$ was executed at least once before, as $n \geq 1$, and observe that $\tilde{r} \in \reg{r}{r'}$. Hence, we identify two situations in the context of the aforementioned execution: (a) either all the newly added pairs corresponding to the assignment $t \assgn s$ complied to the regular structure, or (b) each new pair $[t', s']$ that did not fit the regular pattern was later removed via a subsequent instruction ``$\create{u}$'' or ``$\forget{u}$'' within $p$, with $u$ a prefix of $t'$ or $s'$.

If the case (a) above was satisfied, then, based on the definition of dot-completeness, a pair
\[
(t,s') \in (\tilde{r_1} - t)[t = \tilde{r_1}\slash s - t] - ot,
\]
where 
\[
\tilde{r_1} = \tilde{r}[ot = t]
\]
cannot break the regular pattern of the alias expressions either.
For the case (b) above, all the ``non-well-behaved'' new pairs will be again removed via a subsequent ``$\create{u}$'' or ``$\forget{u}$'' within $p$.

Therefore, the assumption in~(\ref{eq:red-absurdum}) is false, so it holds that:
\[
\overline{r} \exec p = r \exec p^{n+1} \in \reg{r}{r'}.
\]
\end{itemize}
\qed
\end{proof}

\section{Alias calculus in {\K} -- complete specification}
\label{sec:calculus-complete-spec}

In this section we provide the full specification of the alias calculus in {\K}.
Recall that, as a convention, we mark with ($\clubsuit$) the rules that are sound only for non-concurrent contexts, based on Lemma~\ref{lm:reg-expr}.

The following {\K}-rules are straightforward, based on the axioms~(\ref{eq:def-seq-comp})--(\ref{eq:def-assign}) in Section~\ref{sec:ext-inf-expr}. 
Namely, the rule implementing an instruction $p \,;\, q$ simply forces the sequential execution of $p$ and $q$ by positioning $p \curvearrowright q$ at the top of the continuation cell:

\begin{equation}
\label{appeq:k-seq-comp}
\cellUL{p\,;\,q}{p \curvearrowright q }{k}
\end{equation}

Handling $\create x$ and $\forget x$ complies to the associated definitions. Namely, it updates the current alias relation by removing all the pairs having (at least) one element with $x$ as prefix. In addition, it also pops the corresponding instruction from the continuation stack:

\begin{equation}
\label{appeq:k-create-forget}
\begin{array}{lr}
\cellU{r}{r-x}{al}\!\!\cellUL{\create x}{.}{k} \hspace{15pt}
&
\cellU{r}{r-x}{al}\!\!\cellUL{\forget x}{.}{k}
\end{array}
\end{equation}

The assignment rule restores the current alias relation according to its axiom in~(\ref{eq:def-assign}), and removes the assignment instruction from the top of the $k$-cell:

\begin{equation}
\label{appeq:k-assign}
\cellU{r}{(r_{1} - t)[t\,=\,(r_1 \slash s \,-\,t)] - ot}{al}\cellUL{t \assgn s}{.}{k}\,\,\,\,\,
\textnormal{with~} r_{1} = r[ot = t]
\end{equation}

The {\K}-implementation of a $\thenelse{p}{q}$ statement is more sophisticated, as it instruments a stack-based mechanism enabling the computation of the union of alias relations $r \exec p \,\cup\, r\exec q$ in three steps. First, we define the {\K}-rule:
\begin{equation}
\label{appeq:k-then-else-first}
\cellS{r}{al}
\cellUL{\thenelse{p}{q}}{p \mfbox{et} q \mfbox{ee}}{k}
\cellUL{.}{\cellS{r,p}{t}\,\, \cellS{r,q}{e}}{bkt-te}
\end{equation}
saving at the top of the back-tracking stack $\langle - \rangle_{\textnormal{bkt-te}}$ the initial alias relation $r$ to be modified by both $p$ and $q$, via two cells $\langle r, p\rangle_{\textnormal{t}}$ and $\langle r, q\rangle_{\textnormal{e}}$, respectively. Note that the original instruction in the $k$-cell is replaced by a meta-construction marking the end of the executions corresponding to the {\bf then} and {\bf else} branches with $\mfbox{et}$ and $\mfbox{ee}$, respectively. 

Second, whenever the successful execution of $p$ (signaled by $\mfbox{et}$) at the top of the $k$-cell) builds an alias relation $r'$, the execution of $q$ starting with the original relation $r$ is forced by replacing $r'$ with $r$ in $\langle - \rangle_{\textnormal{al}}$, and by positioning $q \mfbox{ee}$ at the top of the $k$-cell. The new alias information after $p$, denoted by $\langle r', p\rangle_{\textnormal{t}}$, is updated in the back-tracking cell: 

\begin{equation}
\label{appeq:k-end-then}
\cellU{r'}{r}{al}
\cellUL{\mfbox{et} q \mfbox{ee}}{q \mfbox{ee}}{k}
{
\begin{tabular}{r@{}c@{}c@{}l}
\small
$\langle\,\,$ & $\cellS{r,p}{t}$ & $\cellS{r,q}{e}$ & $ \,\,\ldots\rangle_{\textnormal{bkt-te}}$\\
\cline{2-2}
& $\cellS{r',p}{t}$ & &\\
\end{tabular}
}
\end{equation}

Eventually, if the successful execution of $q$ (marked by $\mfbox{ee}$ at the top of $\langle-\rangle_{\textnormal k}$) produces an alias relation $r''$, then the final alias information becomes $r' \,\cup\, r''$, where $r'$ is the aliasing after $p$, stored as showed in~(\ref{appeq:k-end-then}). The corresponding back-tracking information is removed from  $\langle - \rangle_{\textnormal{bkt-te}}$, and the next program instruction is enabled in the $k$-cell:

\begin{equation}
\label{appeq:k-end-else}
\cellU{r''}{r'\,\cup\, r''}{al}
\cellUL{\mfbox{ee}}{.}{k}
\cellUL{\cellS{r',p}{t}\,\,\cellS{r,q}{e}}{.}{bkt-te}
\end{equation}

For $\loopend{p}$, we utilize a meta-construction $p \mfbox{l} \loopend{p}$ simulating the set union in~(\ref{eq:def-loop}), and a back-tracking stack $\langle - \rangle_{\textnormal{bkt-l}}$ collecting the alias information obtained after each execution of $p$. Moreover, the {\K}-implementation exploits the result in Lemma~\ref{lm:reg-expr}.
Whenever a ``lasso'' is reached, the infinite rewriting is prevented by resuming the infinite application of $p$ in terms of a sound over-approximating alias relation. The {\K}-rules are as follows.

First, the aforementioned unfolding is performed, and the alias relation before $p$ is stored in the back-tracking cell as $\langle r \rangle_{\textnormal{al-o}} \langle p \rangle_{\textnormal{l}}$:

\begin{equation}
\label{appeq:k-loop-whole}
\cellS{r}{al}
\cellUL{\loopend{p}}{p \mfbox{l} \loopend{p}}{k}
\cellUL{.}{\cellS{r}{al-o}\cellS{p}{l}}{bkt-l}
\end{equation}

If the alias relation $r'$ obtained after the successful execution of $p$ (marked by $\mfbox{l}$ at the top of the continuation) is not a lasso of the aliasing $r$ before $p$ (previously stored in $\langle -\rangle_{\textnormal{bkt-l}}$) then $p$ is constrained to a new execution by becoming the top of the $k$-cell, and $r'$ is memorized for back-tracking:

\begin{equation}
\label{appeq:k-loop-not-lasso}
\cellS{r'}{al}
\cellUL{\mfbox{l} \loopend{p}}{p \mfbox{l} \loopend{p}}{k}
\cellUL{\cellS{r}{al-o}\cellS{p}{l}}{\cellS{r'}{al-o}\cellS{p}{l}}{bkt-l}
\textnormal{~if not~} \lasso{r}{r'}~~(\clubsuit)
\end{equation}

Last, if a lasso is reached after the execution of $p$, then the current aliasing is soundly replaced by a ``regular'' over-approximation $\reg{r}{r'}$, the corresponding back-tracking information is removed from $\langle -\rangle_{\textnormal{bkt-l}}$ and the {\bf loop} instruction is eliminated from the $k$-cell:

\begin{equation}
\label{appeq:k-loop-lasso}
\cellU{r'}{\reg{r}{r'}}{al}\!\!
\cellUL{\mfbox{l} \loopend{p}}{.}{k}\!\!
\cellUL{\cellS{r}{al-o}\cellS{p}{l}}{.}{bkt-l}
\textnormal{~if~} \lasso{r}{r'}~~(\clubsuit)
\end{equation}

For handling function calls such as $\call{f(l)}$ we use a meta-construction $\mid f \mid \!\! \mfbox{f}$. Here $\mid f \mid$ stands for the body of $f$ and $\mfbox{f}$ marks the end of the corresponding execution. Moreover, a stack $\langle - \rangle_{\textnormal{bkt-cf}}$ is utilized in order to store the alias information before each (possibly recursive) call of $f$, with the purpose of identifying the  lassos generated by the (possibly repeated) execution of $f$.
In order to guarantee a sound implementation of (mutually) recursive calls,
both $\mfbox{f}$ and  $\langle - \rangle_{\textnormal{bkt-cf}}$ are parameterized by $f$ -- the name of the function.
An example illustrating this reasoning mechanism is provided in Appendix~\ref{sec:example-k-machinery}.

The first {\K} rule for handling function calls matches the associated axiom in~(\ref{eq:def-qualified-call}): the alias information is set to $r[f^{\bullet}:l]$, whereas the next instructions to be executed are given by $\mid f \mid$. Note that the original aliasing is retained in the (initially empty) back-tracking cell via $\langle r \rangle_{\textnormal{al-o}}$.

\begin{equation}
\label{appeq:k-call-first}
\cellU{r}{r[f^\bullet:l]}{al}
\cellUL{\call{f(l)}}{\mid f \mid \mfbox{f}}{k}
\cellU{.}{\cellS{r}{al-o}}{bkt-cf}
\end{equation}

\begin{remark}
\label{apprem:formal-vs-actual}
Observe that the back-tracking cell does not need to be parameterized by the actual argument list $l$ of $f$. Each such argument is anyways replaced in the current alias relation $r$ by its counterpart in the formal argument list of $f$. In short: $r$ becomes $r[f^\bullet : l]$.
\end{remark}

A successful execution of $\call{f(l)}$ is distinguished by the occurrence of $\mfbox{f}$ at the top of the continuation stack. If this is the case, then the corresponding back-tracking alias information is removed from $\langle - \rangle_{\textnormal{bkt-cf}}$ and the next program instruction (if any) is enabled at the top of the $k$-cell:

\begin{equation}
\label{appeq:k-call-exit}
\cellS{r'}{al}
\cellUL{\mfbox{f}}{.}{k}
\cellUL{\cellS{r}{al-o}}{.}{bkt-cf}
\end{equation}

Recursive calls are treated by means of two {\K}-rules. Note that a recursive context is identified whenever the current program instruction is of shape $\call{f(l)}$ and the associated back-tracking structure is not empty, {\it i.e.}, rule~(\ref{appeq:k-call-first}) was previously applied.
Then, if the recursive call of $f$ when starting with $r$ produces a lasso $r'$, the execution of $f(l)$ is stopped by soundly over-approximating the alias information with $\reg{r}{r'}$, according to Lemma~\ref{lm:reg-expr}, and by removing $\call{f(l)}$ from the $k$-cell:

\begin{equation}
\label{appeq:k-call-lasso}
\cellU{r'}{\reg{r}{r'}}{al}
\cellUL{\call{f(l)}}{.}{k}
\cellSL{\cellS{r}{al-o}}{bkt-cf}
\textnormal{~if~} \lasso{r}{r'}~~(\clubsuit)
\end{equation}

If a lasso is not reached, then the body of $f$ is executed once more, and the current aliasing is pushed to the back-tracking cell:

\begin{equation}
\label{appeq:k-call-not-lasso}
\cellS{r'}{al}
\cellUL{\call{f(l)}}{\mid f \mid \mfbox{f}}{k}
{
\begin{tabular}{r@{}c@{}c@{}l}
\small
$\langle\,\,$ & $.$ & $\cellS{r}{al-o}$ & $ \,\,\ldots\rangle_{\textnormal{bkt-cf}}$\\
\cline{2-2}
& $\cellS{r'}{al-o}$ & &\\
\end{tabular}
}
\textnormal{~if~not~} \lasso{r}{r'}~~(\clubsuit)
\end{equation}

Qualified calls $x.\call{f(l)}$ are handled by two {\K}-rules as follows. First, based on the definition in~(\ref{eq:def-qualified-call}), the ``negative variable'' $x'$ transposing the context of the call to to the context of the caller is distributed to the elements of the initial alias relation $r$, and to $l$ -- the argument list of $f$. Moreover, a meta-construction $\mfbox{qf}$ is utilized in order to mark the end of the qualified call in the continuation cell, similarly to the rule~(\ref{appeq:k-call-first}). The caller is stored in a back-tracking stack $\cellS{.}{\textnormal{bkt-qf}}$ also parameterized by $f$ -- the name of the function. The current instruction in the $k$-cell becomes $\call{f(x'.l)}$, as expected:

\begin{equation}
\label{appeq:k-qcall-first}
\cellU{r}{x'.r}{al}
\cellUL{x. \call{f(l)}}{\call{f(x'.l)} \mfbox{qf}}{k}
\cellU{.}{\cellS{x}{f}}{bkt-qf}
\end{equation}

Second, when the successful termination of the qualified call is signaled by $\mfbox{qf}$ at the top of the $k$-cell, the corresponding stored caller is distributed to the current alias relation and removed from the back-tracking cell. The next instruction in the continuation cell is released by eliminating the top $\mfbox{qf}$: 

\begin{equation}
\label{appeq:k-qcall-exit}
\cellU{r}{x.r}{al}
\cellUL{\mfbox{qf}}{.}{k}
\cellUL{\cellS{x}{f}}{.}{bkt-qf}
\end{equation}

\section{The {\K}-machinery by example}
\label{sec:example-k-machinery}

For an example, in this section we show how the {\K}-machinery developed in Section~\ref{sec:implem-k} can be used in order to extract the alias information for the case of two mutually recursive functions defined as:
\[
\begin{array}{lr}
\begin{array}{l}
f(x) \,\, \{
\,\, x \assgn x.a\,; 
~~\call{g(x)} \,\, \}
\end{array}
&\hspace{40pt}
\begin{array}{l}
g(x) \,\, \{
\,\, x \assgn x.b\,;
~~\call{f(x) \,\, \}}
\end{array}
\end{array}
\]
We assume that $x$ is an object of a class with two fields $a$ and $b$, respectively. We consider a sequential setting.

At first glance it is easy to see that the execution of $\call{f(x)}$, when starting with an empty alias relation $r$, produces the alias expressions:
\begin{equation}
\label{eq:ex-mut-rec-alias}
[x,\,x.(a.b)^*]~~~[x.a,\,x.(a.b)^*.a]~~~[x.b,\,x.(a.b)^*.b]
\end{equation}

The associated reasoning in {\K} is depicted in the figure below. The whole procedure starts with an empty alias relation $r = \emptyset$, and $\call{f(x)}$ in the continuation stack. Then, the corresponding {\K} rules (for handling assignments and function calls) are applied in the natural way.

A lasso is reached after two calls of $f(x)$ that, consequently, determine two calls of $g(x)$ -- identified by $\mfbox{g} \mfbox{f} \mfbox{g} \mfbox{f}$ in the $k$-cell. This triggers the application of rule~(\ref{appeq:k-call-lasso})  enabling the ``regular'' over-approximation as in Lemma~\ref{lm:reg-expr}.

Our example also illustrates the importance of isolating the back-traced alias information in cells of shape $\cellS{.}{\textnormal{bkt-cf}}$ parameterized by the (possibly recursive) function $f$. More explicitly, rule~(\ref{appeq:k-call-lasso}) is soundly applied by identifying the aforementioned lasso based on: the current alias relation $r_4$, the recursive call $f(l)$ at the top of the continuation, and the back-traced aliasing $\cellSL{\cellS{r_2}{\textnormal{al-o}}}{\textnormal{bkt-cf}}$ associated to the previous executions of $f(l)$.

\renewcommand{\arraystretch}{1.5}
\setlength{\tabcolsep}{2pt}

As introduced in~(\ref{eq:def-lasso}), an alias relation $r'$ is a lasso of a relation $r$ whenever there is a one-to-one correspondence between their elements as follows:
\[
[x_1 y_1 z_1, x_2 y_2 z_2] \in r \textnormal{~~iff~~} [x_1 y_1 y_1 z_1, x_2 y_2 y_2 z_2] \in r'.
\]
The current alias relation
\[
r_4 = \{[x, x.a.b.a.b],\, [x.a, x.a.b.a.b.a],\, [x.b, x.a.b.a.b.b]\},
\]
before applying rule~(\ref{appeq:k-call-lasso}), is a lasso of
\[
r_2 = \{[x, x.a.b],\, [x.a, x.a.b.a],\, [x.b, x.a.b.b]\}.
\]
The aforementioned one-to-one correspondence is summarized in the following table:
\[
\footnotesize
\begin{tabular}{r@{}c@{}l|c|c|c|c|c|c}
$[x_1 y_1 z_1, x_2 y_2 z_2] \in r_2$& {~iff~} & $ [x_1 y_1 y_1 z_1, x_2 y_2 y_2 z_2] \in r_4$ & $x_1$ & $y_1$ & $z_1$ & $x_2$ & $y_2$ & $z_2$\\
\hline
$[x, x.a.b] \in r_2$ & {~iff~} & $[x, x.a.b.a.b] \in r_4$ & $x$ & $\varepsilon$ & $\varepsilon$ & $x$ & $a.b$ & $\varepsilon$\\
\hline
$[x.a, x.a.b.a] \in r_2$ & {~iff~} & $[x.a, x.a.b.a.b.a] \in r_4$ & $x$ & $\varepsilon$ & $a$ & $x$ & $a.b$ & $a$\\
\hline
$[x.b, x.a.b.b] \in r_2$ & {~iff~} & $[x.b, x.a.b.a.b.b] \in r_4$ & $x$ & $\varepsilon$ & $b$ & $x$ & $a.b$ & $b$
\end{tabular}
\]
Here $\varepsilon$ stands for the \emph{empty alias expression}.

Moreover, according to rule~(\ref{appeq:k-call-lasso}), the lasso shaped by $r_2$ and $r_4$ also causes the (otherwise infinite) recursive calls to stop, as $\call{f(l)}$ is eliminated from the top of the $k$-cell.
Hence, the rewriting process finishes with
a sound over-approximation $\reg{r_2}{r_4}$ replacing the current alias relation (cf. Lemma~\ref{lm:reg-expr}), defined precisely as in~(\ref{eq:ex-mut-rec-alias}). 

\renewcommand{\arraystretch}{1}

\begin{figure}[!htbp]
\[
\begin{tabular}{rl}
\label{fig:ex-mut-rec-k}
$\cellS{r}{\textnormal{al}}$ & $\cellS{\call{f(x)}}{\textnormal{k}}$\\
$\cellS{.}{\textnormal{bkt-cf}}$ & $\cellS{.}{\textnormal{bkt-cg}}$
\\[1.5ex]
\multicolumn{2}{c}{$\Downarrow$~(\ref{appeq:k-call-first})}\\[1.5ex]
$\cellS{r}{\textnormal{al}}$ & $\cellS{x \assgn x.a;\, \call{g(x)} \mfbox{f}}{\textnormal{k}}$\\
$\cellS{\cellS{r}{\textnormal{al-o}}}{\textnormal{bkt-cf}}$ & $\cellS{.}{\textnormal{bkt-cg}}$
\\[1.5ex]
\multicolumn{2}{c}{$\Downarrow$~(\ref{appeq:k-assign})}\\[1.5ex]
$\cellS{r_1}{\textnormal{al}}$ & $\cellS{\call{g(x)} \mfbox{f}}{\textnormal{k}}$\\
$\cellS{\cellS{r}{\textnormal{al-o}}}{\textnormal{bkt-cf}}$ & $\cellS{.}{\textnormal{bkt-cg}}$\\
\multicolumn{2}{c}{where $r_1 = \{[x, x.a],\, [x.a, x.a.a],\, [x.b, x.a.b]\}$}
\\[1.5ex]
\multicolumn{2}{c}{$\Downarrow$~(\ref{appeq:k-call-not-lasso})}\\[1.5ex]
$\cellS{r_1}{\textnormal{al}}$ & $\cellS{x \assgn x.b;\, \call{f(x)} \mfbox{g} \mfbox{f}}{\textnormal{k}}$\\
$\cellS{\cellS{r}{\textnormal{al-o}}}{\textnormal{bkt-cf}}$ & $\cellS{\cellS{r_1}{\textnormal{al-o}}}{\textnormal{bkt-cg}}$
\\[1.5ex]
\multicolumn{2}{c}{$\Downarrow$~(\ref{appeq:k-assign})}\\[1.5ex]
$\cellS{r_2}{\textnormal{al}}$ & $\cellS{\call{f(x)} \mfbox{g} \mfbox{f}}{\textnormal{k}}$\\
$\cellS{\cellS{r}{\textnormal{al-o}}}{\textnormal{bkt-cf}}$ & $\cellS{\cellS{r_1}{\textnormal{al-o}}}{\textnormal{bkt-cg}}$\\
\multicolumn{2}{c}{where $r_2 = \{[x, x.a.b],\, [x.a, x.a.b.a],\, [x.b, x.a.b.b]\}$}
\\[1.5ex]
\multicolumn{2}{c}{$\Downarrow$~(\ref{appeq:k-call-not-lasso})}\\[1.5ex]
$\cellS{r_2}{\textnormal{al}}$ & $\cellS{ x \assgn x.a;\,\call{g(x)} \mfbox{f} \mfbox{g} \mfbox{f}}{\textnormal{k}}$\\
$\cellS{\cellS{r_2}{\textnormal{al-o}}\,\, \cellS{r}{\textnormal{al-o}}}{\textnormal{bkt-cf}}$ & $\cellS{\cellS{r_1}{\textnormal{al-o}}}{\textnormal{bkt-cg}}$
\\[1.5ex]
\multicolumn{2}{c}{$\Downarrow$~(\ref{appeq:k-assign})}\\[1.5ex]
$\cellS{r_3}{\textnormal{al}}$ & $\cellS{ \call{g(x)} \mfbox{f} \mfbox{g} \mfbox{f}}{\textnormal{k}}$\\
$\cellS{\cellS{r_2}{\textnormal{al-o}}\,\, \cellS{r}{\textnormal{al-o}}}{\textnormal{bkt-cf}}$ & $\cellS{\cellS{r_1}{\textnormal{al-o}}}{\textnormal{bkt-cg}}$\\
\multicolumn{2}{c}{where $r_3 = \{[x, x.a.b.a],\, [x.a, x.a.b.a.a],\, [x.b, x.a.b.a.b]\}$}
\\[1.5ex]
\multicolumn{2}{c}{$\Downarrow$~(\ref{appeq:k-call-not-lasso})}\\[1.5ex]
$\cellS{r_3}{\textnormal{al}}$ & $\cellS{ x\assgn x.b;\,\call{f(x)} \mfbox{g} \mfbox{f} \mfbox{g} \mfbox{f}}{\textnormal{k}}$\\
$\cellS{\cellS{r_2}{\textnormal{al-o}}\,\, \cellS{r}{\textnormal{al-o}}}{\textnormal{bkt-cf}}$ & $\cellS{\cellS{r_3}{\textnormal{al-o}}\,\ \cellS{r_1}{\textnormal{al-o}}}{\textnormal{bkt-cg}}$
\\[1.5ex]
\multicolumn{2}{c}{$\Downarrow$~(\ref{appeq:k-assign})}\\[1.5ex]
$\cellS{r_4}{\textnormal{al}}$ & $\cellS{ \call{f(x)} \mfbox{g} \mfbox{f} \mfbox{g} \mfbox{f}}{\textnormal{k}}$\\
$\cellS{\cellS{r_2}{\textnormal{al-o}}\,\, \cellS{r}{\textnormal{al-o}}}{\textnormal{bkt-cf}}$ & $\cellS{\cellS{r_3}{\textnormal{al-o}}\,\ \cellS{r_1}{\textnormal{al-o}}}{\textnormal{bkt-cg}}$\\
\multicolumn{2}{c}{where $r_4 = \{[x, x.a.b.a.b],\, [x.a, x.a.b.a.b.a],\, [x.b, x.a.b.a.b.b]\}$}
\\[1.5ex]
\multicolumn{2}{c}{$\Downarrow$~(\ref{appeq:k-call-lasso})}\\[1.5ex]
$\cellS{reg(r_2, r_4)}{\textnormal{al}}$ & $\cellS{ \mfbox{g} \mfbox{f} \mfbox{g} \mfbox{f}}{\textnormal{k}}$\\
$\cellS{\cellS{r_2}{\textnormal{al-o}}\,\, \cellS{r}{\textnormal{al-o}}}{\textnormal{bkt-cf}}$ & $\cellS{\cellS{r_3}{\textnormal{al-o}}\,\ \cellS{r_1}{\textnormal{al-o}}}{\textnormal{bkt-cg}}$\\
\\[1.5ex]
\multicolumn{2}{c}{$\Downarrow$~(*)(\ref{appeq:k-call-exit})}\\[1.5ex]
\multicolumn{2}{c}{
$\cellS{\{[x, x.(a.b)^*],\,[x.a, x.(a.b)^*.a],\, [x.b, x.(a.b)^*.b]  \}}{\textnormal{al}} \cellS{.}{\textnormal{k}}
\cellS{.}{\textnormal{bkt-cf}} \cellS{.}{\textnormal{bkt-cg}}$}
\end{tabular}
\]
\caption{Aliasing and mutual recursion in {\K}.}
\end{figure}

\newpage
\section{Example of aliasing in Maude}
\label{app:Maude-aliasing}

The intermediate class-based representation (in {\small \verb+aliasing-linked_list.maude+}) corresponding to the example
\[
\begin{array}{l}
x\,:=y;\\
\loopend{x\,:=x.next}
\end{array}
\]
is given as:
{\small
\begin{verbatim}
(class 'LINKED_LIST_TEST 
    create  { 'make }
    (
        procedure { 'ANY } 'make ( nil ) 
            require True 
            local 
                (  'x : [?, . , 'NODE] ;  'y : [?, . , 'NODE] ;  )
            do
                ( 
                assign ('x, 'y);
                until False loop ( assign ('x, 'x . 'next(nil)) ; ) end ;
                )
               ensure True
               rescue nil
           end ;
    )
	
    invariant  True 
end) ;

class 'NODE
    create  {'make}
    (
    attribute { 'ANY } 'next : [?, . , 'NODE] ;
    )
	
    invariant True 
end ;
\end{verbatim}
}
As can be seen from the code above, the syntax enables expressing Eiffel-like properties of classes by using assertions s.a. preconditions (introduced by the keyword \verb+require+), postconditions (through the keyword \verb+ensure+) and class invariants.

The ``entry point'' of the program corresponds to the function {\small \verb+'make+} in the (main) class {\small \verb+'LINKED_LIST_TEST+} and is set via:
{
\small
\begin{verbatim}
settings('LINKED_LIST_TEST, 'make, false, aliasing-on) .
\end{verbatim}
}
Observe that the flag for performing the alias analysis is switched to ``on''.
In {\small \verb+'LINKED_LIST_TEST+}, two local variables $x$ and $y$ of type {\small \verb+NODE+} are declared as running on the current processor ({\small \verb+.+}), {\it i.e.}, they are not \emph{separate}. The instruction {\small \verb+assign('x, 'y)+}, for instance, corresponds to the assignment $x\,:=y$. The class defining a {\small \verb+NODE+} structure (in a linked list) simply consists of a (non-separate) field {\small \verb+'next+} of type {\small \verb+NODE+}.

We run the example by executing the command:
{
\small
\begin{verbatim}
> maude SCOOP.maude ..\examples\aliasing-linked_list.maude
\end{verbatim}
}

The relevant parts of the corresponding Maude output are as follows:
{
\small
\begin{verbatim}
                     \||||||||||||||||||/
                   --- Welcome to Maude ---
                     /||||||||||||||||||\
            Maude 2.6 built: Mar 31 2011 23:36:02
            Copyright 1997-2010 SRI International
[...]
==========================================
rewrite [100] in SYSTEM : 
[...]
{0}proc(1) ::
until False loop
  assign('x, 'x . 'next(nil)) ;
end ;
[...],
100,
aliasing-on
({['x ; 'y . 'next . 'next . 'next . 'next .
    'next . 'next . 'next . 'next . 'next . 'next . 'next . 'next .
    'next . 'next . 'next . 'next . 'next . 'next . 'next . 'next .
    'next . 'next . 'next . 'next . 'next . 'next . 'next . 'next .
    'next . 'next . 'next . 'next . 'next . 'next . 'next . 'next .
    'next . 'next . 'next . 'next . 'next . 'next]} U
 {['x . 'next ; 'y . 'next . 'next . 'next . 'next . 'next . 'next . 
    'next . 'next . 'next . 'next . 'next . 'next . 'next . 'next .
    'next . 'next .  'next . 'next . 'next . 'next . 'next . 'next . 
    'next . 'next . 'next . 'next . 'next . 'next . 'next . 'next .
    'next . 'next . 'next . 'next . 'next . 'next . 'next . 'next .
    'next . 'next . 'next . 'next . 'next]})
[...]
state
  [...]
  heap [...]
  store [...]
end
\end{verbatim}
}

In short, after $100$ rewriting steps, the current processor {\small \verb+{0}proc(1)+} has the execution corresponding to $\loopend{x\,:=x.next}$ on top of its instruction stack, and the aliasing information contains (the dot-complete closure of) the relation $\{[x, y.next^{42}]\}$. Moreover, the output displays the contents of the current system state, by providing information on the \emph{heap} and \emph{store}, as formalized in~\cite{DBLP:conf/acsd/MorandiSNM13}.
\end{document}